\documentclass[conference]{IEEEtran}
\IEEEoverridecommandlockouts

\usepackage{amsmath,amssymb,amsfonts}
\usepackage{tabularx,booktabs}
\usepackage{mathtools}
\usepackage{algorithmic}
\usepackage{graphicx}
\usepackage{textcomp}
\usepackage{xcolor}
\usepackage{subcaption}
\usepackage{nccmath}
\usepackage{hyperref}
\usepackage{dsfont}
\def\BibTeX{{\rm B\kern-.05em{\sc i\kern-.025em b}\kern-.08em
    T\kern-.1667em\lower.7ex\hbox{E}\kern-.125emX}}
\newcolumntype{C}{>{\centering\arraybackslash}X} 
\newcolumntype{b}{>{\hsize=2.3\hsize}X}
\usepackage{amsthm}
\usepackage{bbm}
\usepackage{csquotes}
\usepackage{chngcntr}
\usepackage{apptools}

\usepackage{cite}

\theoremstyle{plain}
\newtheorem{theorem}{Theorem}
\newtheorem{lemma}{Lemma}
\newtheorem*{lemma*}{Lemma}
\newtheorem{corollary}{Corollary}

\theoremstyle{definition}
\newtheorem{definition}{Definition}

\theoremstyle{remark}

\newtheorem{remark}{Remark}
\newtheorem{example}{Example}

\newcommand{\X}{\mathcal{X}}
\newcommand{\Y}{\mathcal{Y}}

\newcommand{\E}{\mathbb{E}}

\newcommand{\Pm}{\mathcal{P}}
\newcommand{\Q}{\mathcal{Q}}
\newcommand{\F}{\mathcal{F}}
\newcommand{\W}{\mathcal{W}}

\makeatletter
\newcounter{labelcnt}
\renewcommand{\thelabelcnt}{(\alph{labelcnt})}
\newcommand{\setlabel}[1]{%
  \refstepcounter{labelcnt}\ltx@label{lbl:#1}%
  {\text{\upshape\thelabelcnt}}%
}

\makeatother

\def\arxiv{1}

\begin{document}

\title{Lower-bounds on the Bayesian Risk in Estimation Procedures via $f$--Divergences}

\author{
\IEEEauthorblockN{Adrien Vandenbroucque, Amedeo Roberto Esposito, Michael Gastpar}
\IEEEauthorblockA{\textit{School of Computer and Communication Sciences} \\
EPFL, Lausanne, Switzerland\\
adrien.vandenbroucque@alumni.epfl.ch, 
\{amedeo.esposito, michael.gastpar\}@epfl.ch}
}

\maketitle

\begin{abstract}
  We consider the problem of parameter estimation in a Bayesian setting and propose a general lower-bound that includes part of the family of $f$-Divergences. The results are then applied to specific settings of interest and compared to other notable results in the literature. In particular, 
  we show that the known bounds using Mutual Information can be improved by using, for example, Maximal Leakage, Hellinger divergence, or generalizations of the Hockey-Stick divergence.
  
\end{abstract}

\section{Introduction}
In this work\if\arxiv0\footnote{A more detailed version of this work including all the proofs can be found in \cite{arxiv_version_of_this_paper}. Moreover, the code for the experiments can be found at \url{https://github.com/Adirlou/f\_divergences\_lower\_bounds}}\fi \,we consider the problem of parameter estimation in a Bayesian setting.
The connection between said problem and information measures has been established multiple times over the years \cite{zhang,bayesRiskRaginsky,CalmonSDPIHockeyStick}.
Here we further develop the perspective undertaken in \cite{bayesRiskRaginsky} and in \cite{bayesRiskIalpha}.
Similarly to \cite{bayesRiskRaginsky} and \cite{bayesRiskIalpha} we will look at the problem through an information-theoretic lens and we will thus treat the parameter to be estimated as a message sent through a channel. The family of bounds one can derive in this framework generally give rise to two objects:
\begin{itemize}
    \item a measure of information (Shannon's Mutual Information was employed in \cite{bayesRiskRaginsky}, Sibson's $\alpha$-Mutual Information in \cite{bayesRiskIalpha}, Hockey-Stick Divergence in \cite{CalmonSDPIHockeyStick}, etc.);
    \item a small-ball probability;
\end{itemize}
The main advantage of this is that both terms can be rendered independent of the specific choice of the estimator, which in turns renders these lower-bounds quite general.
Our main focus will not be on asymptotic results but rather on finite sample lower-bounds. In particular, we will expand upon \cite{bayesRiskIalpha}, utilizing the same approach but focusing on $f-$Divergences rather than on Sibson's Mutual Information.
\section{Background and definitions}
\begin{definition}
Given a function $f:\mathcal{X}\to\mathcal{Y}$, 
	the Legendre-Fenchel transform of $f$ is defined as
	\begin{equation}
    f^\star(x^\star) = \sup_{x\in\X} \langle x^\star , x\rangle - f(x),
\end{equation}
\end{definition}
where $\langle  x^\star,x \rangle$ denotes the natural pairing between a space $\X$ and its topological dual $\X^\star$, \textit{i.e.},  $\langle x^\star,x \rangle= x^\star(x).$
Given a function $f$, $f^\star$ is guaranteed to be lower semi-continuous and convex. If $f$ is convex and lower semi-continuous then $f=f^{\star\star}|_{\X}$ (the restriction of $f^{\star\star}$ on $\X$ agrees with $f$).

\subsection{$f-$Divergences}\label{fDivergence}
A straightforward generalization of the KL-Divergence can be obtained by considering a generic convex function $f:\mathbb{R}\to \mathbb{R}$, usually with the simple constraint that $f(1)=0$. 
\begin{definition}

	Let $(\Omega,\F,\Pm),(\Omega,\F,\Q)$ be two probability spaces. Let $f:\mathbb{R}\to \mathbb{R}$ be a convex function such that $f(1)=0$. Consider a measure $\mu$ such that $\Pm\ll\mu$ and $\Q\ll\mu$ (\textit{i.e.}, $\Pm$ and $\Q$ are absolutely continuous with respect to $\mu$). Denoting with $p,q$ the densities of the measures with respect to $\mu$, the $f-$Divergence of $\Pm$ from $\Q$ is defined as follows:
	\begin{align}
	\label{defFDiv}
	D_f(\Pm\|\Q)=\int q f\left(\frac{p}{q}\right) d\mu.
	\end{align}    
\end{definition}
Note that $f-$divergences are independent from the choice of the dominating measure $\mu$~\cite{fDiv1}. When absolute continuity between $\Pm,\Q$ holds, denoted with $\Pm\ll\Q$ one retrieves the following~\cite{fDiv1}:
\begin{equation}
    D_f(\Pm\|\Q)= \int f\left(\frac{d\Pm}{d\Q}\right)d\Q.
\end{equation}
This generalization includes the KL divergence (by simply setting $f(t)=t\log(t)$), but it also includes:
\begin{itemize}
	\item Total Variation distance, with $f(t)=\frac12|t-1|$;
	\item Hellinger distance, with $f(t)=(\sqrt{t}-1)^2$;
	\item Pearson $\chi^2$-divergence, with $f(t)=(t-1)^2$.
\end{itemize}
In particular, in this paper, we will be interested in two families of divergences. The first family, also known as Hellinger Divergences, is typically characterized by a parameter $p>0$. More precisely, we are referring to the $f$--Divergences that stem from $f_p(t)= \frac{t^p-1}{p-1}$ and that will be denoted as follows:
\begin{equation}
    \mathcal{H}_p(\Pm\|\Q)= D_{f_p}(\Pm\|\Q).
\end{equation}
The second family we consider is characterized by two parameters, namely $\beta >0$ and $\gamma\geq \beta$, and arise from the parametric family of functions $f_{\beta, \gamma}(t) = \max\{0, \beta t-\gamma\}$. We denote it as:
\begin{equation}
    E_{\beta, \gamma}(\Pm\|\Q)= D_{f_{\beta, \gamma}}(\Pm\|\Q).
\end{equation}
For the case $\beta=1$, one retrieves the family of so-called $E_{\gamma}$--Divergences \cite[Eq. (47)]{sason_f_divergences}.

Much like $f$--Divergences, a generalization of Shannon's Mutual Information, denoted in the literature as 
$f$--Mutual Information, can be defined starting from $f$--Divergences as follows:
\begin{definition}\label{def:fMI}
    Let $X$ and $Y$ be two random variables jointly distributed according to $\Pm_{XY}$ over a measurable space $(\X\times\Y, \F_{XY})$. 
 Let  $(\X,\F_{X},\Pm_{X}),(\Y,\F_{Y},\Pm_Y)$ be the corresponding probability spaces induced by the marginals.  Let $f:\mathbb{R}\to \mathbb{R}$ be a convex function such that $f(1)=0$. The $f$--Mutual Information between $X$ and $Y$ is defined as:
		\begin{equation}
	I_f(X,Y)=D_f(\Pm_{XY}\|\Pm_X\Pm_Y).
	\end{equation}
\end{definition}
If $f$ is strictly convex at $1$ and satisfies $f(1)=0$, then  $I_f(X,Y)=0$ if and only if $X$ and $Y$ are independent~\cite[Theorem 5]{fDiv1}. Choosing $f(t)=t\log t$, one recovers the Mutual Information. With a slight abuse of notation, we will denote $f$--Mutual Informations with the same symbols used to characterize the corresponding divergences, \textit{e.g.}, $\mathcal{H}_p(X,Y)=\mathcal{H}_p(\Pm_{XY}\|\Pm_X\Pm_Y)$ will represent the $f_p$--Mutual Information, while $E_{\beta, \gamma}(X,Y)=E_{\beta, \gamma}(\Pm_{XY}\|\Pm_X\Pm_Y)$ will represent the $f_{\beta, \gamma}$--Mutual Information.

\subsection{Problem Setting - the Bayesian framework} \label{sec:bayesianFramework}
Let $\mathcal{W}$ denote the parameter space and assume that we have access to a prior distribution over this space $\mathcal{P}_W$. Suppose then that we observe $W$ through the family of distributions $\mathcal{P}= \{ \mathcal{P}_{X|W=w}: w\in\mathcal{W} \}.$ Given a function $\phi:\mathcal{X}\to\mathcal{W}$ one can then estimate $W$ from $X\sim \mathcal{P}_{X|W}$ via $\phi(X)=\hat{W}$. Let us denote with $\ell:\mathcal{W}\times\mathcal{W}\to \mathbb{R}^+$ a loss function, the Bayesian risk is defined as:
\begin{equation}
    R = \inf_\phi\mathbb{E}[\ell(W,\phi(X)] = \inf_\phi\mathbb{E}[\ell(W,\hat{W})].\label{risk}
\end{equation}
Our purpose will be to lower-bound $R$ using the tools described in the previous section. To this end, we will be using a simple Markov's inequality approach: \textit{i.e.}, for every estimator $\phi$ and $\rho \geq 0$, one can do the following
\begin{equation}
    \mathbb{E}[\ell(W,\hat{W})] \geq \rho\left(P_{W\hat{W}}(\ell(W,\hat{W})\geq \rho)\right). \label{markov}
\end{equation}
With further manipulations we can actually relate $\mathbb{P}(\ell(W,\hat{W})\geq \rho)$ to the information-measures described before and some function $\psi$ of $P_WP_{\hat{W}}(\ell(W,\hat{W})\geq \rho)$ (the measure of $\{\ell(W,\hat{W})\geq \rho\}$ under the product of the marginals $P_WP_{\hat{W}}$).
Let us denote $P_W P_{\hat{W}}(\ell(W,\hat{W})\leq \rho)=L_W(\hat{W},\rho)$.
In some cases, this $\psi$ will lead us to considering
the so-called small-ball probability  \begin{equation}L_W(\rho)= \sup_{\hat{w}\in\mathcal{\hat{W}}} L_W(\hat{w},\rho) =  \sup_{\hat{w}\in\mathcal{\hat{W}}}\mathbb{P}(\ell(W,\hat{w})\leq \rho).\label{smallBall}\end{equation}
The purpose is to render both of these quantities independent of $\phi$, granting us the tools to provide general lower-bounds on the risk $R$.
\subsection{Related
Works} A survey of early works in this area, mainly focusing on asymptotic settings, can be found in \cite{han}. More recent but important advances are instead due to \cite{zhang, nipsHideAndSeek}.
Closely connected to this work is \cite{bayesRiskRaginsky}. The approach is quite similar, with the main difference that we employ a family of bounds involving a variety of divergences while ~\cite{bayesRiskRaginsky} relies solely on Mutual Information and the Kullback-Leibler Divergence.~\cite{bayesRiskIalpha} focuses on Sibson's $\alpha$-Mutual Information, and \cite{CalmonSDPIHockeyStick} uses the $E_\gamma$-Divergence. A similar approach was also undertaken in~\cite{bayesRiskFInformativity}. The authors focused on the notion of $f-$informativity (cf.~\cite{fInformativity}) and leveraged the data processing inequality similarly to~\cite[Theorem 3]{fullVersionGeneralization}. In particular, $f-$informativities are more general than the $f-$Mutual Informations considered in this work (cf. Definition~\ref{def:fMI}) and they can potentially lead to tighter results. The technique used to provide lower-bounds on the Bayesian risk for general non-negative losses (cf.~\cite[Section 4]{bayesRiskFInformativity}) is, however, different. It is unclear whether the results provided in this work are equivalent (or weaker) with respect to those obtained in~\cite{bayesRiskFInformativity}. 
\section{The lower bounds}
Let us start with our main result and then show how it is connected to the Bayesian Risk.
\begin{theorem}
\label{theorem:f_divergences_bound}
Consider the Bayesian framework described in Sec.~\ref{sec:bayesianFramework}. Let $f:[0, +\infty) \to \mathbb{R}$ be an increasing convex function such that $f(1)=0$ and suppose that the generalized inverse, defined as $f^{-1}(y) = \inf\{t\geq 0 : f(t) > y\}$, exists. Then the following must hold for every $\rho>0$ and every estimator $\hat{W}$:
\begin{multline}
    \E[\ell(W, \hat{W})] \geq \rho\Bigg(1- L_W(\hat{W},\rho)\cdot \\ f^{-1}\left(\frac{I_f(W, \hat{W})+(1-L_W(\hat{W},\rho))f^\star(0)}{L_W(\hat{W},\rho)}\right)\Bigg). \label{fDivLowerBound}
\end{multline}
Moreover, if $f^\star(0) \leq 0$, the bound simplifies to
\begin{align}
    \E[\ell(W, \hat{W})] \geq \rho\left(1- L_W(\hat{W},\rho)\cdot f^{-1}\left(\frac{I_f(W, \hat{W})}{L_W(\hat{W},\rho)}\right)\right). \label{fDivLowerBoundSimplified}
\end{align}
\end{theorem}
\begin{proof}
To prove the statement we use \cite[Theorem 3]{fullVersionGeneralization}. In our notation, it states that for every function $f$ with the desired properties, we have 
\begin{align}
    P_{W\hat{W}}(\ell(W, \hat{W})&\leq \rho) \leq L_W(\hat{W}, \rho) \cdot\\& f^{-1}\left(\frac{I_f(W, \hat{W})+(1-L_W(\hat{W},\rho))f^\star(0)}{L_W(\hat{W},\rho)}\right).  \label{lastStepProofMainThm}
\end{align}
In particular when $f^\star(0)\leq 0$, the bound reduces to
\begin{align}
    P_{W\hat{W}}(\ell(W, \hat{W})\leq \rho) \leq L_W(\hat{W}, \rho) \cdot f^{-1}\left(\frac{I_f(W, \hat{W})}{L_W(\hat{W},\rho)}\right).
\end{align}
Rewriting $P_{W\hat{W}}(\ell(W, \hat{W})\geq \rho)$ as $1-P_{W\hat{W}}(\ell(W, \hat{W})\leq \rho)$ and combining this with Equations \eqref{markov} and \eqref{lastStepProofMainThm} concludes the proof.
\end{proof}
In order to provide a lower-bound on the Bayesian Risk, one needs to render the right-hand side of Equations \eqref{fDivLowerBound} (or \eqref{fDivLowerBoundSimplified}) independent of $\hat{W}=\phi(X)$ and, in order to do that, one needs to render independent of $\hat{W}$:
\begin{enumerate}
    \item The information-measure, \textit{e.g.}, through the data-processing inequality $I_f(W, \hat{W}) \leq I_f(W, X)$; \label{ineq1}
    \item The quantity $L_W(\hat{W}, \rho)$, that can be easily upper-bounded in the following way: $L_W(\hat{W}, \rho) \leq \sup_{\hat{w}} L_W(\hat{w}, \rho) = L_W(\rho)$. \label{ineq2}
\end{enumerate}
For simplicity, consider Equation \eqref{fDivLowerBoundSimplified} and introduce the following object
\begin{equation}
    G_f(I_f, L_W) := L_W(\hat{W},\rho)\cdot f^{-1}\left(\frac{I_f(W, \hat{W})}{L_W(\hat{W},\rho)}\right).
\end{equation}
To use the two inequalities just stated above in items \ref{ineq1}) and \ref{ineq2}), one thus needs that for a given choice of $f$, $G_f(I_f, L_W)$ is increasing in $I_f$ for a given value of $L_W$ and vice-versa. This allows us to further lower-bound \eqref{fDivLowerBoundSimplified} and render the quantity independent of the specific choice of $\phi$. Hence, starting from \eqref{risk} one can provide a lower-bound on the risk $R$ that is independent of $\phi$.
Let us now look at some specific choices of $f$ such that $G_f$ satisfies the desired properties and for which a bound on the Bayesian risk can indeed be retrieved. 
\begin{corollary}\label{lowerBoundHellinger}
Consider the Bayesian framework described in Sec. \ref{sec:bayesianFramework}. The following must hold for every $p>1$ and $\rho>0$:
\begin{equation}
     R \geq \rho\left(1- L_W(\rho)^{\frac{p-1}{p}}\cdot  \left((p-1) \mathcal{H}_p(W,X) + 1\right) ^\frac{1}{p} \right).
\end{equation}
\end{corollary}
\begin{proof}
Since $f(x) = \frac{x^p-1}{p-1}$, we have that $f^\star(0) = \sup_{x \geq 0} (- f(x)) = \frac{1}{p-1}$ and $f^{-1}(t) = ((p-1)t + 1)^\frac1p$.

For every estimator $\hat{W}$,
\begin{align}
    &L_W(\hat{W},\rho)\cdot f^{-1}\left(\frac{I_f(W, \hat{W})+(1-L_W(\hat{W},\rho))f^\star(0)}{L_W(\hat{W},\rho)}\right)  
    \\&= L_W(\hat{W}, \rho)\left(\frac{(p-1)\mathcal{H}_p(W, \hat{W})+1}{L_W(\hat{W}, \rho)}\right)^{\frac1p} \\&= L_W(\hat{W}, \rho)^{\frac{p-1}{p}}\left((p-1)\mathcal{H}_p(W, \hat{W})+1\right)^{\frac1p} \\&\leq L_W(\rho)^{\frac{p-1}{p}}\left((p-1)\mathcal{H}_p(W, X)+1\right)^{\frac1p}, \label{lastStepProof1}
\end{align}
where in \eqref{lastStepProof1} we used the data-processing inequality for $f$--divergences. Using \eqref{lastStepProof1} with Theorem \ref{theorem:f_divergences_bound}, we retrieve that for every estimator $\hat{W}$
\begin{equation}
    \mathbb{E}[\ell(W, \hat{W})] \geq\rho\!\left(1-L_W(\rho)^{\frac{p-1}{p}}\left((p-1)\mathcal{H}_p(W, X)+1\right)^{\frac1p}\right). \label{lowerBoundTemp}
\end{equation}
Since the right-hand side of \eqref{lowerBoundTemp} is independent of $\hat{W}=\phi(X)$ one can use it to lower-bound the risk $R$. 
\end{proof}
Restricting the choice of $f$ to this family of polynomials we can thus state the following lower-bound on the risk:
\begin{equation}\label{equation:lowerBoundHellinger}
    R \geq \sup_{\rho > 0} \sup_{p > 1}\rho\!\left(\!1- L_W(\rho)^{\frac{p-1}{p}}\!\!\cdot  \left((p-1) \mathcal{H}_p(W,\hat{W}) + 1\right)\! ^\frac{1}{p}\! \right).
\end{equation}
\begin{remark}
Using the one-to-one mapping connecting Hellinger divergences and R\'enyi's $\alpha-$Divergence \cite[Eq. (30)]{sason_f_divergences}, the bound above can be re-written as follows:
\begin{align}
     R \geq \sup_{\rho >0} \sup_{\alpha>1} \rho&\bigg(1- L_W(\rho)^{\frac{\alpha-1}{\alpha}}\cdot 
     \notag \\  &\exp\left(\frac{\alpha-1}{\alpha} D_\alpha(P_{W\hat{W}}\|P_WP_{\hat{W}})\right) \bigg).
\end{align}
\end{remark}
In addition, given the generality of Theorem \ref{theorem:f_divergences_bound} we can also recover other notable results present in the literature (cf.~\cite[Remark 1]{CalmonSDPIHockeyStick}) through the following:
\begin{corollary}\label{lowerBoundHockeyStick}
Consider the Bayesian framework described in Sec. \ref{sec:bayesianFramework}. The following must hold for every $\beta > 0$, $\gamma\geq \beta$, and $\rho>0$:
\begin{equation}\label{equation:hockey_stick_general_bound}
     R \geq \rho\left(1 - \frac{E_{\beta, \gamma}(W,\hat{W}) + \gamma L_W(\rho)}{\beta}\right).
\end{equation}
\end{corollary}
\begin{proof}
We take the same approach as in Corollary \ref{lowerBoundHellinger}. Let $f(x) = \max\{0, \beta x-\gamma\}$, consequently one has that $f^\star(0) = \sup_{x \geq 0} (- f(x)) = 0$ and that the generalized inverse corresponds to $f^{-1}(t) = \frac{t+\gamma}{\beta}$. Using Theorem \ref{theorem:f_divergences_bound}, along with the fact that $f^\star(0) \leq 0$ we have that for every estimator $\hat{W}$,
\begin{align}
    \E[\ell(W, \hat{W})] &\geq  
    \rho\left(1-\frac{E_{\beta,\gamma}(W, \hat{W}) + \gamma L_W(\hat{W}, \rho)}{\beta}\right) \\ 
    &\geq \rho\left(1-\frac{E_{\beta, \gamma}(W, X) + \gamma L_W(\rho)}{\beta}\right). \label{equation:e_gamma}
\end{align}
Since \eqref{equation:e_gamma} is independent of $\hat{W}=\phi(X)$ one can use it to lower-bound the risk $R$.
\end{proof}
We thus retrieve the following lower-bound on the risk:
\begin{equation}\label{equation:lowerBoundHockeyStick}
     R \geq \sup_{\rho>0}\sup_{\beta>0, \gamma \geq \beta}\rho\left(1 - \frac{E_{\beta, \gamma}(W,\hat{W}) + \gamma L_{W}(\rho)}{\beta}\right).
\end{equation}
\begin{remark}
Note that setting $\beta = 1$ \eqref{equation:hockey_stick_general_bound} recovers the result in \cite[Remark 1]{CalmonSDPIHockeyStick}.
In fact, by introducing an additional degree of freedom through the $\beta$ parameter in Equation \eqref{equation:lowerBoundHockeyStick}, the resulting lower-bound can only be tighter than \cite[Remark 1]{CalmonSDPIHockeyStick}.
\end{remark}
\section{Examples}
In this section we apply Corollaries \ref{lowerBoundHellinger} and \ref{lowerBoundHockeyStick} to two classical estimation settings. The resulting lower-bounds are then compared with those obtained in
\cite{bayesRiskIalpha} involving Sibson's $\alpha$-Mutual Information and Maximal Leakage and with those in \cite{bayesRiskRaginsky} involving Shannon's Mutual Information and Maximal Leakage.

Ultimately, for each example, we would like to compare the tightest versions of our bounds, which are given by Equation \eqref{equation:lowerBoundHellinger} for the $\mathcal{H}_p$--Divergence and \eqref{equation:lowerBoundHockeyStick} for the $E_{\beta,\gamma}$--Divergence. However, since their computations involve a maximization problem over some parameters ($p$ or $\beta, \gamma$) that we cannot analytically solve, we compute these lower-bounds only for specific values of the parameters. The choice of parameters we use might seem arbitrary but it correctly captures the behavior of the bounds. Indeed, experiments show that when solving the maximization over $p$ or $\beta, \gamma$ (\textit{e.g.}, through the \texttt{scipy.optimize.minimize} function from the Python library SciPy) the same behaviors are observed, like Figure \ref{fig:bernoulli_best_params} shows in the context of Example \ref{bernoulliBias}.
\begin{figure}
    \centering
    \includegraphics[width=0.48\textwidth]{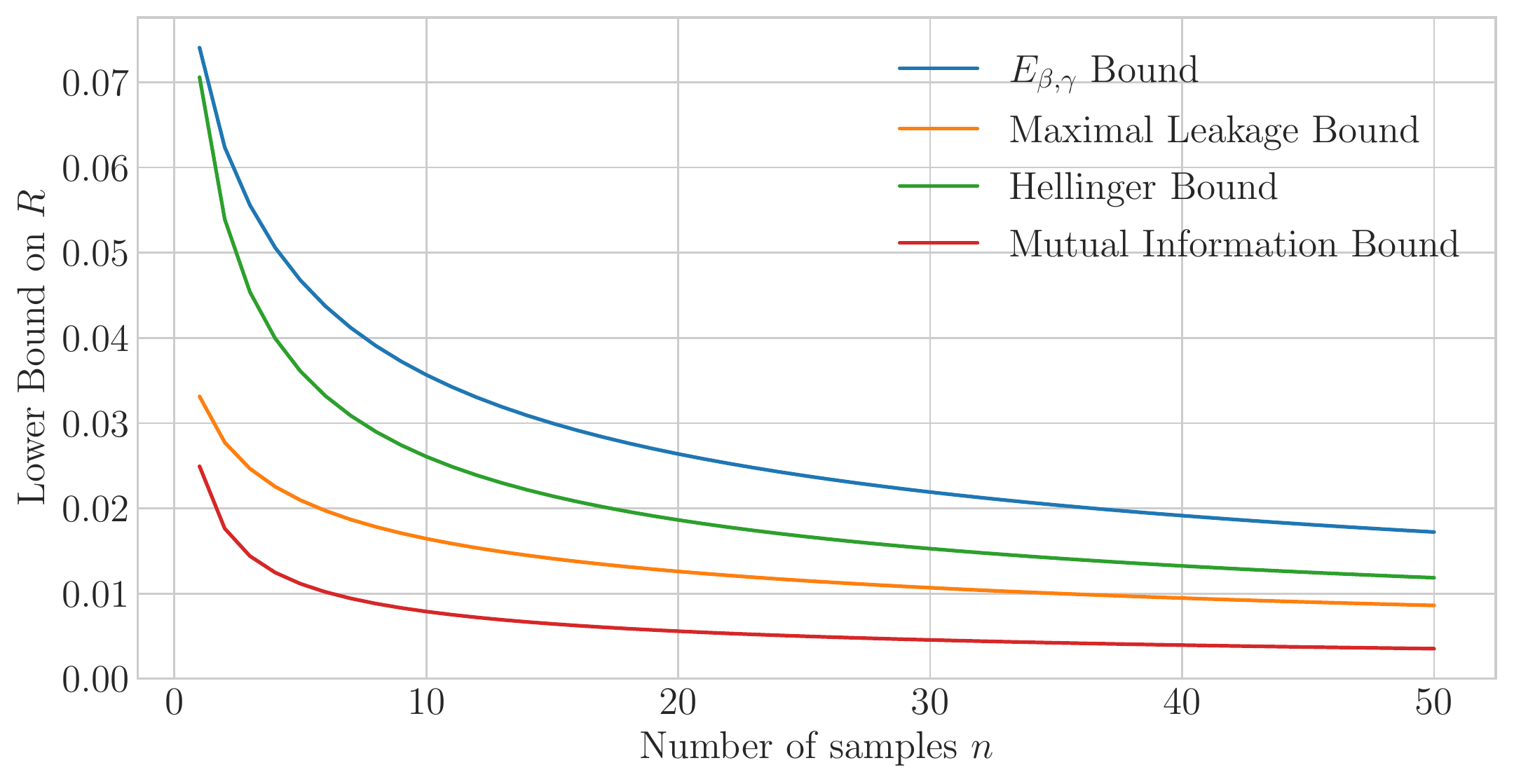}
    \caption{Setting: Example \ref{bernoulliBias}. Comparison between the largest lower-bounds one can retrieve for different information measures in Example \ref{bernoulliBias}: that is between \eqref{equation:lowerBoundHockeyStick},  \eqref{equation:lower_bound_hellinger_bernoulli}, \cite[Eq. (16)]{bayesRiskIalpha} and \cite[Corollary 2, Eq. (19)]{bayesRiskRaginsky}. The quantities are analytically maximized over $\rho$ (cf.\if\arxiv1 Appendix \ref{appendix:optimization_rho}\else \cite[Appendix A]{arxiv_version_of_this_paper}\fi) and numerically optimized over, respectively, $p>1$, $\beta>0$, and $\gamma\geq \beta$.
    }
    \label{fig:bernoulli_best_params}
\end{figure}

\subsection{Example 1: Bernoulli Bias Estimation}\label{sec:ex1}
\begin{example}\label{bernoulliBias}
Suppose that $W \sim U[0,1]$ and that for each $i\in[n]$, $X_i|\{W\!=\!w\} \sim \text{Ber}(w)$. Also, assume that $\ell(w,\hat{w}) = |w-\hat{w}|$.
\end{example}

We first provide a closed-form expression of the lower-bound resulting from Corollary \ref{lowerBoundHellinger} for a specific choice of $p$ which enables to match the upper-bound up to a constant factor. In fact in general, the tightest bound in this family comes from Equation \eqref{equation:lowerBoundHellinger} and can, in this example, be stated as follows:
\begin{equation}\label{equation:lower_bound_hellinger_bernoulli}
    R \geq \sup_{\rho > 0} \sup_{p > 1}\rho\!\left(\!1- (2\rho)^{\frac{p-1}{p}}\!\!\cdot  \left((p-1) \mathcal{H}_p(W,X^n) + 1\right)\! ^\frac{1}{p}\! \right).
\end{equation}
The value of $\mathcal{H}_p(W,X^n)$ for this setting is expressed in the following Lemma. \begin{lemma}\label{lemma:hellinger_bernoulli}
Consider the setting described in Example \ref{bernoulliBias}. Then for every $p>1$,
\begin{align}
    (p-1)(&\mathcal{H}_p(W,X^n)+1) \notag\\&= (n+1)^{p-1}\sum_{k=0}^n \binom{n}{k}^p\frac{\Gamma(kp+1)\Gamma((n-k)p+1)}{\Gamma(np+2)}.
\end{align}
In particular with $p=2$, one recovers:
\begin{equation}
    \chi^2(W,X^n)+1 = \frac{n+1}{2n+1}\cdot \frac{4^n}{\binom{2n}{n}}\leq \frac{16\sqrt{\pi n}}{21}. \label{H2Bernoulli}
\end{equation}
\end{lemma}
\begin{proof}
\if\arxiv1
    See Appendix \ref{appendix:proof_lemma_hellinger_bernoulli}.
\else 
    The proof is omitted due to space constraints but can be found in \cite[Appendix B]{arxiv_version_of_this_paper}.
\fi
\end{proof}
\begin{corollary}\label{corollary:bernoulli_lower_bound_closed_form}
Consider the setting described in Example \ref{bernoulliBias}. The Bayesian risk is lower-bounded by
\begin{equation}\label{equation:bernoulli_lower_bound_closed_form}
    R \geq \frac{7}{72\sqrt{\pi n}}.
\end{equation}
\end{corollary}

\begin{proof}

Let $p=2$ in Corollary \ref{lowerBoundHellinger} along with $L_W(
\rho)\leq 2\rho$, one has that 
\begin{equation}
    R \geq \sup_{\rho>0}\rho \left(1-\sqrt{2\rho(\chi^2(W, X^n)+1)}\right).
\end{equation}
Solving the maximization over $\rho$ (cf.\if\arxiv1 Appendix \ref{appendix:optimization_rho}) \else \cite[Appendix A]{arxiv_version_of_this_paper}\fi) and using \eqref{H2Bernoulli} we conclude that
\begin{align}
    R &\geq \frac{2}{27}\cdot \frac{1}{\chi^2(W, X^n)+1} \geq \frac{7}{72\sqrt{\pi n}}.
\end{align}
\end{proof}
Notice that \eqref{equation:bernoulli_lower_bound_closed_form} matches the upper-bound up to a constant, and tightens the result in \cite[Corollary 2]{bayesRiskRaginsky} while not requiring that $n \to \infty$.
\begin{remark}
As mentioned in previous proof, Stirling's approximation yields $(\chi^2(W, X^n)+1) \sim \frac{\sqrt{\pi n}}{2}$ when $n$ is large. This implies that for $n$ large one can show that
$R \gtrsim \frac{4}{27\sqrt{\pi n}}$, thus leading to a slight improvement over \eqref{equation:bernoulli_lower_bound_closed_form}.
\end{remark}
Similarly, one can do the same steps used to retrieve Corollary \ref{corollary:bernoulli_lower_bound_closed_form}, but this time using the $E_{\beta,\gamma}$--Divergence instead of the $\mathcal{H}_p$--Divergence. In particular, for the case $\beta=0.75$ and $\gamma=2.2$, Eq. \eqref{equation:hockey_stick_general_bound} in this example can be expressed as
\begin{align}
    R &\geq \sup_{\rho>0}\rho\left(1-\frac{4}{3}\left(E_{0.75, 2.2}(W, X^n) + 4.4\rho\right)\right)\\
    &= 
    \frac{5(0.75-E_{0.75, 2.2}(W, X^n))^2}{66}.\label{equation:bernoulli_hockey_stick}
\end{align}

A direct comparison between the bounds we provide and those already present in the literature can be seen in Figure \ref{fig:bernoulli_specific_params}. The lower-bounds are computed as a function of the number of samples $n$, which we consider to be in the range $\{1, \dots ,50\}$. The figure shows that all the divergences we considered in this work provide a larger (and thus, better) lower-bound on the Bayesian risk when compared with results that stem from using Shannon's Mutual Information (cf. \cite[Corollary 2]{bayesRiskRaginsky}). In particular, the lower-bound involving the $E_{\beta, \gamma}$--Mutual Information represents the largest among the ones we consider. Given the lack of a closed-form expression for $E_{\beta,\gamma}$ in this example the quantities \eqref{equation:bernoulli_hockey_stick} along with (\!\!\cite[Corollary 2, Eq. (19)]{bayesRiskRaginsky} and \cite[Eq. (16)]{bayesRiskIalpha}) and \eqref{equation:bernoulli_lower_bound_closed_form} are computed numerically.
\begin{figure}
    \centering
    \includegraphics[width=0.48\textwidth]{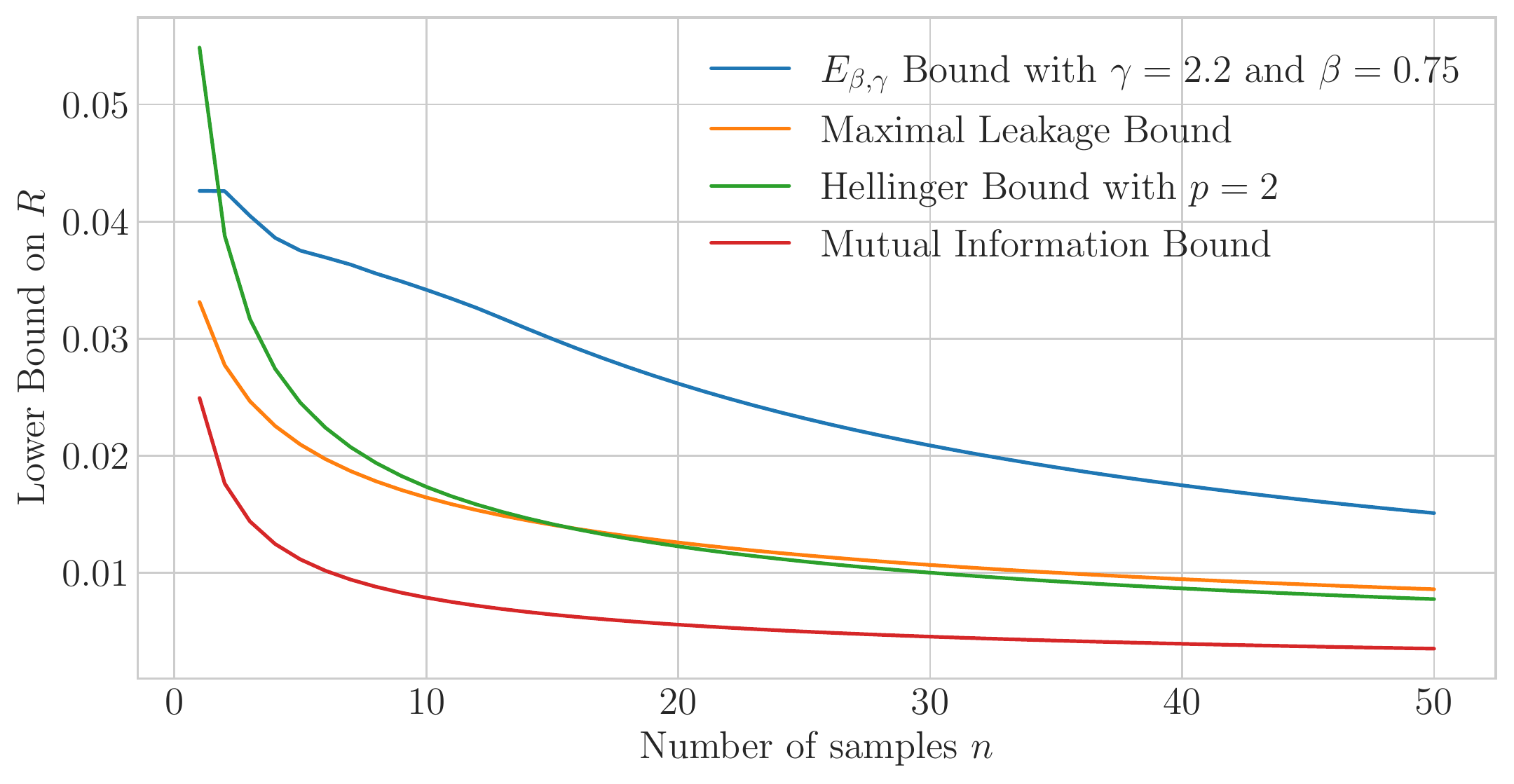}
    \caption{Setting: Example \ref{bernoulliBias}. The picture shows the behaviour of \eqref{equation:bernoulli_lower_bound_closed_form}, \eqref{equation:bernoulli_hockey_stick}, \cite[Eq. (16)]{bayesRiskIalpha}, and\cite[Corollary 2, Eq. (19)]{bayesRiskRaginsky} as a function of $n$. The values of $E_{0.75, 2.2}(W, X^n)$ for each $n$ are computed numerically. Here, unlike in Figure \ref{fig:bernoulli_best_params} where parameters are optimized, the values are fixed to $\gamma=2.2, \beta=0.75$ and $p=2$.}
    \label{fig:bernoulli_specific_params}
\end{figure}

\subsection{Gaussian prior with Gaussian noise in $d$ dimensions}
\begin{example}\label{gaussianPriorGaussianNoise}
Assume that $W\sim N(0,\sigma^2_W)$ and that for $i\in[n]$, $X_i = W + Z_i$ where $Z_i \sim N(0,\sigma^2)$. Assume also that the loss is s.t. $\ell(w,\hat{w}) = |w-\hat{w}|$.
\end{example}
Using the estimator $\hat{W} = \mathbb{E}[W|\bar{X}]$ with $\bar{X} = \frac{1}{n}\sum_{i=1}^n X_i \sim \mathcal{N}(0, \frac{\sigma^2}{n})$, one has that $R \leq \sqrt{\sigma_W^2/\left(1+n\frac{\sigma_W^2}{\sigma^2}\right)}$. Moreover, the small-ball probability can be upper-bounded as follows
\begin{align}
    L_W(\rho) 
   &\leq 
   \left(\sup_{w\in\mathbb{R}} P_W(w)\right)\left(\int_{-\rho}^{\rho} 1du\right)=\frac{2\rho}{\sqrt{2\pi \sigma_W^2}}.
\end{align}
Once again the largest lower bound on the risk, in the family of bounds provided by Corollary \ref{lowerBoundHellinger}, can be expressed as follows
\begin{equation}
    R\!\geq\!\sup_{\rho > 0} \sup_{p > 1}\rho\!\left(\!\!1- \left(\!\!\frac{2\rho}{\sqrt{2\pi\sigma_W^2}}\!\!\right)^{\frac{p-1}{p}}\!\!\!\!\left(\!(p-1)\mathcal{H}_p(W,X^n) \!+\!1\!\right)\!^\frac{1}{p}\!\!\right). \label{equation:gaussian_lower_bound}
\end{equation}
To compute the Hellinger information, we make use of the following lemma:
\begin{lemma}\label{lemma:hellinger_gaussian}
Let $W \sim \mathcal{N}(0, \sigma_W^2I_d)$ and $Z\sim \mathcal{N}(0, \sigma^2I_d)$ be two Gaussian random variables, where $I_d$ denotes the $d\times d$ identity matrix. Moreover, let $X = W+Z$ and $p>1$. Then 
\begin{equation}
    (p-1)(\mathcal{H}_p(W,X)+1) = \left(\frac{\left(1+\frac{\sigma_W^2}{\sigma^2}\right)^p}{1 + (2-p)p\frac{\sigma_W^2}{\sigma^2}}\right)^\frac{d}{2}.
\end{equation}

In particular, with $p=3/2$ and $d=1$, one recovers:
\begin{equation}\label{equation:hellinger_gaussian_32}
    \frac{1}{2}(\mathcal{H}_{3/2}(W, X)+1) =  \sqrt{\frac{\left(1+\frac{\sigma_W^2}{\sigma^2}\right)^\frac32}{1 + \frac{3\sigma_W^2}{4\sigma^2}}}.
\end{equation}
\end{lemma}
\begin{proof}
\if\arxiv1
    See Appendix \ref{appendix:proof_lemma_hellinger}.
\else 
    The proof is omitted due to space constraints but can be found in \cite[Appendix C]{arxiv_version_of_this_paper}.
\fi
\end{proof}
Setting $p=3/2$ in \eqref{equation:gaussian_lower_bound} leads to the following result:
\begin{corollary}\label{corollary:gaussian_lower_bound_closed_form}
Consider the setting described in Example \ref{gaussianPriorGaussianNoise}. The Bayesian risk is lower-bounded by
\begin{equation}\label{equation:gaussian_lower_bound_closed_form}
    R \geq \frac{81\sqrt{2\pi}}{2048}\sqrt{\frac{\sigma_W^2}{1+n\frac{\sigma_W^2}{\sigma^2}}}.
\end{equation}
\end{corollary}
\begin{proof} 
Given that $\bar{X}$ is a sufficient statistic we have that $\mathcal{H}_p(W,X^n) = \mathcal{H}_p(W,\bar{X})$. 
Plugging this choice of $\bar{X}$ in \eqref{equation:hellinger_gaussian_32}, substituting in \eqref{equation:gaussian_lower_bound}, and then optimizing over $\rho$ (cf. \if\arxiv1 Eq. \eqref{rhoStar} \else \cite[Eq. (45)]{arxiv_version_of_this_paper}\fi), yields the statement after some algebraic manipulations.
\end{proof}
Note that \eqref{equation:gaussian_lower_bound_closed_form} matches the upper-bound up to a constant factor, and provides a strengthening of the bounds obtained in \cite[Corollary 1]{bayesRiskRaginsky}.
One can, as in Example \ref{bernoulliBias}, repeat the analysis with the $f_{\beta, \gamma}$--Divergence instead of the $f_p$--Divergence. In particular for the case $\beta=0.75$ and $\gamma=2.2$, Equation \eqref{equation:hockey_stick_general_bound} in this example can be expressed as
\begin{align}
    R &\geq \sup_{\rho>0}\rho\left(1-\frac{4}{3}\left(E_{0.75, 2.2}(W, X^n) + \frac{4.4\rho}{\sqrt{2\pi\sigma_W^2}}\right)\right)\\
    &=\frac{5\sqrt{2\pi\sigma_W^2}(0.75-E_{0.75, 2.2}(W, X^n))^2}{66},\label{equation:gaussian_hockey_stick}
\end{align}
where the optimization over $\rho$ stems from\if\arxiv1 Appendix A\else\cite[Appendix A]{arxiv_version_of_this_paper}\fi.
\begin{figure}
    \centering
    \includegraphics[width=0.48\textwidth]{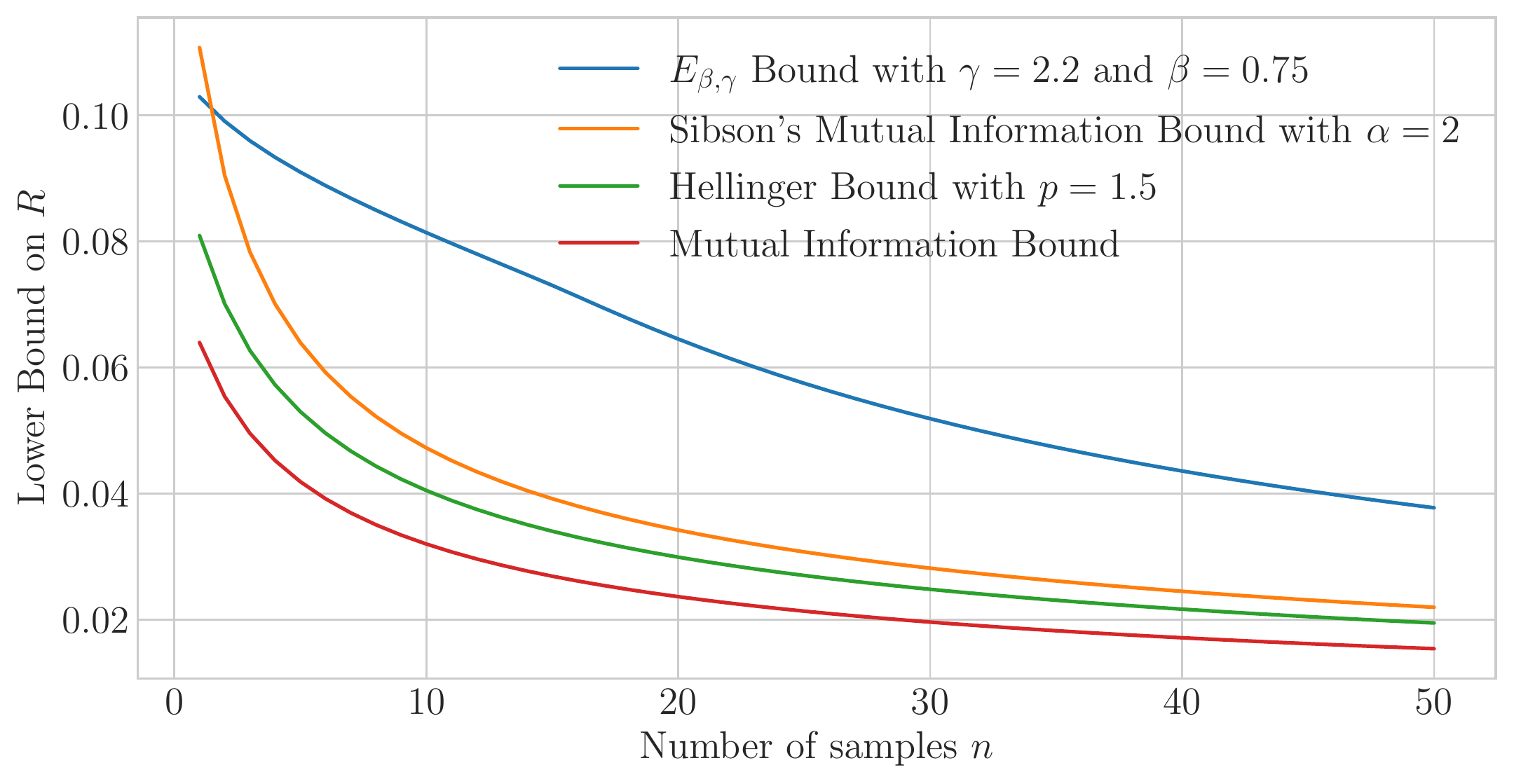}
    \caption{Setting: Example \ref{gaussianPriorGaussianNoise} with $\sigma_W^2=1$ and $\sigma^2=2$. The picture shows the behaviour of \eqref{equation:gaussian_lower_bound_closed_form}, \eqref{equation:gaussian_hockey_stick}, \cite[Eq. (21)]{bayesRiskIalpha}, and \cite[Corollary 1, Eq. (16)]{bayesRiskRaginsky} as a function of $n$. The values of $E_{0.75, 2.2}(W, X^n)$ for each $n$ are computed numerically. Here, the values of the parameters are fixed to $\gamma=2.2, \beta=0.75, \alpha=2$ and $p=1.5$.
    }
    \label{fig:gaussian_best_params}
\end{figure}

Similarly to Example \ref{bernoulliBias}, we numerically evaluate \eqref{equation:gaussian_hockey_stick} and compare it with \cite[Corollary 1, Eq. (16)]{bayesRiskRaginsky}, \cite[Eq. (21)]{bayesRiskIalpha} (with $\alpha=2$), and \eqref{equation:gaussian_lower_bound_closed_form}. Figure \ref{fig:gaussian_best_params} shows the resulting lower-bounds as a function of the number of samples $n$. One can observe similar behaviors when comparing with the results from previous example: the bounds retrieved through the $\mathcal{H}_p$-- and $E_{\beta, \gamma}$--Divergences are able to both improve on the lower-bound relying on Shannon's Mutual Information. Once again, Equation \eqref{equation:lowerBoundHockeyStick} gives the largest lower-bound in this example, while Sibson's $\alpha$-Mutual Information is still able to provide a stronger result than \eqref{equation:lowerBoundHellinger}.
\if\arxiv1
\appendix

\subsection{Maximization over $\rho$}\label{appendix:optimization_rho}
In the two examples considered, one can notice that the lower-bounds resulting from Corollaries \ref{lowerBoundHellinger} and \ref{lowerBoundHockeyStick} have the following form
\begin{equation} \label{equation:rho_equation}
\sup_{\rho>0} \rho(1-c\rho^t-b),
\end{equation}
for some $c, t, b\geq 0$. Letting $ h(\rho):=\rho(1-c\rho^t-b)$, the optimal value is found by setting $h^\prime(\rho_\star) = 0$, which yields
\begin{equation}
    1-(t+1)c\rho_\star^t-b=0 \iff \rho_\star = \left(\frac{1-b}{(t+1)c}\right)^{\frac{1}{t}}. 
\end{equation}
Since $h^{\prime\prime}(\rho_\star) = -t(t+1)c\rho_\star^{t-1} \leq 0$, this ensures $\rho_\star$ is a maximum. Substituting $\rho^\star$ back in \eqref{equation:rho_equation}, we find
\begin{equation}
    \sup_{\rho>0} \rho(1-c\rho^t-b)=\frac{t}{c^{\frac1t}}\left(\frac{1-b}{t+1}\right)^{1+\frac1t}.\label{rhoStar}
\end{equation}

\subsection{Proof of Lemma \ref{lemma:hellinger_bernoulli}}\label{appendix:proof_lemma_hellinger_bernoulli}
In order to prove Lemma \ref{lemma:hellinger_bernoulli}, let us introduce a technical lemma which will be useful in subsequent computations.
\begin{lemma}[\!\!\protect{\cite[Eq. (5.39), p.187]{concrete_mathematics}}]\label{lemma:combinatorial_sum}
Let $n \geq 0$ be a positive integer. Then
\begin{equation}
    \sum_{k=0}^n \binom{2k}{k}\binom{2(n-k)}{n-k} = 4^n.
\end{equation}
\end{lemma}
We can now move on and prove Lemma \ref{lemma:hellinger_bernoulli} which we restate here for reference.
\begin{lemma*}
Consider the setting described in Example \ref{bernoulliBias} \textit{i.e.}, $W \sim U[0,1]$ and  $X_i|\{W\!=\!w\} \sim \text{Ber}(w)$ for each $i\in[n]$. Then for every $p>1$,
\begin{align}
    (p-1)(&\mathcal{H}_p(W,X^n)+1) \notag \\&= (n+1)^{p-1}\sum_{k=0}^n \binom{n}{k}^p\frac{\Gamma(kp+1)\Gamma((n-k)p+1)}{\Gamma(np+2)}.\notag
\end{align}
\end{lemma*}
\begin{proof}
In this specific setting, one has that $P_{X^n|W=w}(x^n) = w^k(1-w)^{(n-k)}$ where $k=\sum_{i=1}^n x_i$, \textit{i.e.}, the hamming weight of $x^n$. As per assumption $P_W(w)=\mathds{1}\{0\leq w\leq 1\}$ and consequently one has that $P_{W|X^n=x^n}(w)= (n+1)\binom{n}{k}(1-w)^{n-k}w^k$. Thus we can compute
\begin{align}
    (p-1)\mathcal{H}_p(W,X^n)+1 &= \\ \sum_{x^n \in \{0, 1\}^n} P_{X^n}(x^n)\int_0^1 P_W(w)\left(\frac{P_{W|X^n=x^n}(w)}{P_W(w)}\right)^pdw &= \label{equation:hellinger_first_equality}\\ \sum_{k=0}^n \frac{1}{n+1}\int_0^1 \left((n+1)\binom{n}{k}w^k(1-w)^{(n-k)}\right)^p dw &= \\ (n+1)^{p-1}\sum_{k=0}^n \binom{n}{k}^p\frac{\Gamma(kp+1)\Gamma((n-k)p+1)}{\Gamma(np+2)},\label{equation:hellinger_computation}
\end{align}
where \eqref{equation:hellinger_first_equality} follows from the definition of Hellinger divergence and \eqref{equation:hellinger_computation} uses the identity relating the Beta function with the Gamma function:
\begin{equation}
    \mathrm{Beta}(x, y) = \int_0^1 t^{x-1}(1-t)^{y-1}dt = \frac{\Gamma(x)\Gamma(y)}{\Gamma(x+y)}.
\end{equation}
If $p=2$ one has that:
\begin{align}
    \chi^2(W,X^n)+1 &= (n+1)\sum_{k=0}^n \binom{n}{k}^2\frac{(2k)!(2(n-k))!}{(2n+1)!}\\&=\frac{n+1}{(2n+1)} \sum_{k=0}^n \frac{(n!)^2(2k)!(2(n-k))!}{(k!)^2((n-k)!)^2(2n)!}\\&= \frac{n+1}{(2n+1)\binom{2n}{n}} \sum_{k=0}^n \binom{2k}{k}\binom{2(n-k)}{n-k} \\&= \frac{n+1}{2n+1}\cdot \frac{4^n}{\binom{2n}{n}}\label{equation:chi_square_before_ub}\\&\leq \frac{2}{3}\cdot \frac{8\sqrt{\pi n}}{7}\label{equation:chi_square_after_ub}\\&= \frac{16\sqrt{\pi n}}{21},\label{equation:bernoulli_chi_square}
\end{align}
where \eqref{equation:chi_square_before_ub} follows from Lemma \ref{lemma:combinatorial_sum}. To obtain \eqref{equation:chi_square_after_ub}, we use $\frac{n+1}{2n+1} \leq \frac{2}{3}$ for $n\geq1$ and Stirling's approximation to get $\binom{2n}{n} \sim \frac{4^n}{\sqrt{\pi n}}$ and retrieve $\binom{2n}{n} \geq \frac{8}{7} \cdot \frac{4^n}{\sqrt{\pi n}}$ for $n\geq 1$.
\end{proof}
\subsection{Proof of Lemma \ref{lemma:hellinger_gaussian}} \label{appendix:proof_lemma_hellinger}
Let us re-state the result for ease of reference.
\begin{lemma*}
Let $W \sim \mathcal{N}(0, \sigma_W^2I_d)$ and $Z\sim \mathcal{N}(0, \sigma^2I_d)$ be two Gaussian random variables, where $I_d$ denotes the $d\times d$ identity matrix. Moreover, let $X = W+Z$ and $p>1$. Then 
\begin{equation*}
    (p-1)(\mathcal{H}_p(W,X)+1) = \left(\frac{\left(1+\frac{\sigma_W^2}{\sigma^2}\right)^p}{1 + (2-p)p\frac{\sigma_W^2}{\sigma^2}}\right)^\frac{d}{2}.
\end{equation*}

In particular, with $p=3/2$ and $d=1$, one recovers:
\begin{equation*}
    \frac{1}{2}(\mathcal{H}_{3/2}(W, X)+1) =  \sqrt{\frac{\left(1+\frac{\sigma_W^2}{\sigma^2}\right)^\frac32}{1 + \frac{3\sigma_W^2}{4\sigma^2}}}.
\end{equation*}
\end{lemma*}
\begin{proof}
First, note that $X|\{W=w\}\sim\mathcal{N}(w, \sigma^2I_d)$. Since the Hellinger information of order $p$ is defined as $\mathcal{H}_p(W,X) = \mathbb{E}_{P_{W}P_{X}}\left[f\left(\frac{dP_{WX}}{dP_{W}P_{X}}\right)\right]$ with $f(t) = \frac{t^p-1}{p-1}$, we have that
\begin{align}
    (p-1)\mathcal{H}_p&(W, X)+1\notag \\
    &=\int_{\mathbb{R}^d} \int_{\mathbb{R}^d} P_{W}(w)P_{X}(x) \left(\frac{P_{X|W=w}(x)}{P_{X}(x)}\right)^p dwdx \\
    &=\int_{\mathbb{R}^d} P_{X}(x)^{1-p} \int_{\mathbb{R}^d} P_{W}(w) P_{X|W=w}(x)^p dwdx.\label{equation:hellinger_integral_lemma2}
\end{align}
Let us denote the inner-most  integral in \eqref{equation:hellinger_integral_lemma2} as $G_p(x)$. One has that:
\begin{align}
    G_p(x):&= \int_{\mathbb{R}^d}P_{{W}}({w})P_{{X}|{W}={w}}({x})^{p} d{w} \\
    &= \left(\frac{(2\pi\sigma^2)^{-p}}{2\pi \sigma_W^2}\right)^{\frac{d}{2}} \int_{\mathbb{R}^d} e^{-\frac{\|{w}\|_2^2}{2\sigma_W^2} - \frac{p \|{w}-{x}\|_2^2}{2\sigma^2}}d{w}.\label{equation:G_p}
\end{align}
Let $I_p(x):=\int_{\mathbb{R}^d} e^{-\frac{\|{w}\|_2^2}{2\sigma_W^2} - \frac{p \|{w}-{x}\|_2^2}{2\sigma^2}}d{w}$ (and thus, $G_p(x)=\left(\frac{(2\pi\sigma^2)^{-p}}{2\pi \sigma_W^2}\right)^{\frac{d}{2}} I_p(x)$) one has
\begin{align}
    I_p(x)&=\int_{\mathbb{R}^d} e^{-\frac{1}{2\sigma^2} \left(p\|{x}\|_2^2-2p{x}^\top {w} + \left(\frac{\sigma^2}{\sigma_W^2}+p\right)\|{w}\|_2^2\right)}d{w}\\
    &= e^{\frac{-p\cdot \|{x}\|_2^2}{2\sigma^2}}\int_{\mathbb{R}^d} e^{-\frac{1}{2\sigma^2} \left(-2p{x}^\top {w} + \left(\frac{\sigma^2}{\sigma_W^2}+p\right)\|{w}\|_2^2\right)}d{w}\label{equation:I_p}
\end{align}
Let us now add and subtract 
$c\|x\|_2^2$ with $c=-p\left(1+p\frac{\sigma_W^2}{\sigma^2}\right)^{-1}$ in the exponent in Equation \eqref{equation:I_p}:
\begin{align}
    I_p(x) &= e^{{\frac{c\|{x}\|_2^2}{2\sigma^2}}}\int_{\mathbb{R}^d} e^{-\frac{\frac{\sigma^2}{\sigma_W^2}+p}{2\sigma^2} \left(\left\|{w}-\sqrt{\frac{p+c}{\frac{\sigma^2}{\sigma_W^2}+p}}{x}\right\|_2^2\right)}d{w} \\
    &= \exp\left(-{\frac{p\cdot \|{x}\|_2^2}{2\sigma^2\left(1+p\frac{\sigma_W^2}{\sigma^2}\right)}}\right) \left(2\pi \frac{\sigma^2}{\frac{\sigma^2}{\sigma_W^2}+p}\right)^{\frac{d}{2}}\label{equation:I_pCompleted}.
\end{align}
Substituting \eqref{equation:I_pCompleted} in \eqref{equation:G_p} gives
\begin{equation}
    G_p(x)=\frac{1}{(2\pi\sigma^2)^{\frac{dp}{2}}} e^{-\frac{p\|{x}\|_2^2}{2\left(\sigma^2+p\sigma_W^2\right)}} \left(1+p\frac{\sigma_W^2}{\sigma^2}\right)^{-\frac{d}{2}}.\label{equation:G_p_value}
\end{equation}
Finally, if we plug in \eqref{equation:G_p_value} in  \eqref{equation:hellinger_integral_lemma2}, we retrieve that:
\begin{align}
    &(p-1)\mathcal{H}_p(W, X)+1 \notag \\
    &= \int_{\mathbb{R}^d} P_{{X}}({x})^{1-p} \frac{1}{(2\pi \sigma^2)^{\frac{dp}{2}}} e^{-\frac{p\|{x}\|_2^2}{2\left(\sigma^2+p\sigma_W^2\right)}} \left(1+p\frac{\sigma_W^2}{\sigma^2}\right)^{-\frac{d}{2}}d{x}\\
    &= \frac{\left(1+\frac{\sigma_W^2}{\sigma^2}\right)^{\frac{d(p-1)}{2}}}{(2\pi \sigma^2)^{\frac{d}{2}}\left(1+p\frac{\sigma_W^2}{\sigma^2}\right)^{\frac{d}{2}}}\int_{\mathbb{R}^d} e^{\frac{(p-1)\|{x}\|_2^2}{2\left(\sigma^2+\sigma_W^2\right)}-\frac{p\|{x}\|_2^2}{2\left(\sigma^2+p\sigma_W^2\right)}} d{x} \\
    &= \frac{\left(1+\frac{\sigma_W^2}{\sigma^2}\right)^{\frac{d(p-1)}{2}}}{(2\pi \sigma^2)^{\frac{d}{2}}\left(1+p\frac{\sigma_W^2}{\sigma^2}\right)^{\frac{d}{2}}}\int_{\mathbb{R}^d} e^{-\frac{\|{x}\|_2^2}{2}\left(\frac{1-p}{\sigma^2+\sigma_W^2}+\frac{p}{\sigma^2+p\sigma_W^2}\right)} d{x} \\
    &= \frac{\left(1+\frac{\sigma_W^2}{\sigma^2}\right)^{\frac{d(p-1)}{2}}}{(2\pi \sigma^2)^{\frac{d}{2}}\left(1+p\frac{\sigma_W^2}{\sigma^2}\right)^{\frac{d}{2}}} \left(\frac{2\pi}{\frac{1-p}{\sigma^2+\sigma_W^2}+\frac{p}{\sigma^2+p\sigma_W^2}}\right)^{\frac{d}{2}} \\
    &= \frac{\left(1+\frac{\sigma_W^2}{\sigma^2}\right)^{\frac{d(p-1)}{2}}}{\left(\sigma^2+p\sigma_W^2\right)^{\frac{d}{2}}} \left(\frac{1}{\frac{1-p}{\sigma^2+\sigma_W^2}+\frac{p}{\sigma^2+p\sigma_W^2}}\right)^{\frac{d}{2}} \\
    &= \left(\frac{\left(1+\frac{\sigma_W^2}{\sigma^2}\right)^{p-1}}{\frac{(1-p)\left(\sigma^2+p\sigma_W^2\right)}{\sigma^2+\sigma_W^2}+p} \right)^{\frac{d}{2}} \\
    &=  \left(\frac{\left(1+\frac{\sigma_W^2}{\sigma^2}\right)^p}{1 + (2-p)p\frac{\sigma_W^2}{\sigma^2}} \right)^{\frac{d}{2}},
\end{align}
which concludes the proof.
\end{proof}
\fi
\bibliographystyle{IEEEtran}
\bibliography{sample}

\end{document}